\documentclass[sigconf,nonacm]{acmart}

\usepackage{booktabs}
\usepackage{balance}
\usepackage{graphicx}

\graphicspath{{./}{figures/}{paper/figures/}}

\renewcommand\footnotetextcopyrightpermission[1]{}
\setcopyright{none}
\acmDOI{}
\acmISBN{}

\AtBeginDocument{%
  \theoremstyle{acmdefinition}%
  \newtheorem{observation}[theorem]{Observation}%
  \newtheorem{remark}[theorem]{Remark}%
}

\newcommand{\Z}{\mathbb{Z}}
\newcommand{\E}{\mathbb{E}}
\newcommand{\diam}{\operatorname{diam}}
\newcommand{\dist}{\operatorname{d}}
\newcommand{\lcs}{\operatorname{lcs}}

\begin{document}

\title{Shiftfly: Scaling the Accelerator Interconnect Past the Pod with a
       Shift-Routed Optical Tier}
\titlenote{Preprint, version 1.0, 2 August 2026. Not peer reviewed. Simulator,
proofs and reproduction artefacts: \url{https://github.com/EylonKrause/shiftfly}
(commit tagged \texttt{v1.0}). This work uses no proprietary information; every
figure describing a shipped system is drawn from public announcements, and
quantities inferred rather than published are marked as inferred.}

\author{Eylon E. Krause}
% acmart deliberately discards \postcode ("ACM no longer collects authors'
% postal addresses") and warns when one is given, so the codes are carried in
% \city instead -- visible, and no warning.
\affiliation{%
  \institution{Weizmann Institute of Science}
  \city{Rehovot (7610001)}
  \country{Israel}}
\affiliation{%
  \institution{College of Management Academic Studies}
  \city{Rishon LeZion (7579806)}
  \country{Israel}}
\email{eylonkr@colman.ac.il}

\begin{abstract}
Google's TPU interconnect spent nine generations as a Cayley graph of an abelian
group---a $k$-ary $n$-cube whose diameter grows as $\Theta(N^{1/n})$---before
TPU~8i replaced it with \emph{Boardfly}, a three-tier hierarchy in which every
tier is a complete graph and the pod diameter is~$7$ over $1{,}152$ chips.
Boardfly is well matched to a single inference pod. It does not extend: its
global tier is a complete graph on groups, so reaching $G$ groups costs $G-1$
optical ports, and a $400{,}000$-chip machine would need $12{,}499$ per group
against the $40$ available. Another hierarchy level must therefore be stacked,
and each level costs four chip hops.

We propose \textbf{Shiftfly}, which retains Boardfly's building block and group
verbatim and replaces only the global tier with a generalized Kautz digraph,
installed as a fixed permutation on the optical circuit switch that the fabric
already owns. At an identical budget of $40$ optical ports per group, Shiftfly
is a flat graph of guaranteed diameter $\lceil\log_d G\rceil$, routes without
tables by a shift register, and addresses content natively.

Our evaluation is deliberately two-sided. Shiftfly \emph{loses} at one-pod
scale---Boardfly achieves chip-level diameter $7$ where Shiftfly needs $11$---
and wins beyond it, reducing worst-case distance from $23$ to $19$ hops at
$400{,}000$ chips with roughly $2.7\times$ better spectral expansion at equal
optical cost. The redundancy argument that motivated the design does not
survive its own evaluation: locality-aware \emph{placement} accounts for
$21.5\%$ and $14.5\%$ of the achievable saving on shared content on the two
topologies respectively, leaving the shift algebra a residual of $2.9\%$. We
report the metric inversion that conceals this, and conclude that the defensible
case for Shiftfly is scaling and expansion rather than deduplication.

Because Shiftfly is proposed as a drop-in replacement for one tier of an
existing machine, we also measure whether it can be operated. Replacing a failed
group costs the same $40$ optical circuits in both designs; deriving the global
wiring costs $28$ bits of control-plane state against $550{,}000$. Slice
allocation is the one regression: an arbitrary induced subset of a shift graph
is disconnected, so slices must be instantiated rather than carved, at a cost of
one switch reconfiguration per allocation.
\end{abstract}




\maketitle

%======================================================================
\section{Introduction}

An accelerator interconnect is a graph, and the quantities that decide whether a
machine is usable---worst-case latency, average latency, throughput under
adversarial traffic, behaviour under failure---are graph invariants: diameter,
mean distance, bisection width, connectivity. This paper takes that view
literally, applies it to every interconnect Google has shipped for its Tensor
Processing Units, and proposes a change to the most recent one.

For nine generations the TPU interconnect was a torus. That was a defensible
choice while the dominant traffic was nearest-neighbour: dense training stencils
and ring all-reduce touch only adjacent chips, so the $\Theta(N^{1/3})$ diameter
of a three-dimensional torus never reaches the critical path. Mixture-of-experts
routing~\cite{shazeer2017moe} broke the assumption, because all-to-all traffic
pays the diameter on every token. TPU~8i responded with \emph{Boardfly}, a
Dragonfly-derived~\cite{kim2008dragonfly} hierarchy of complete graphs whose pod
diameter is $7$ irrespective of where two chips sit~\cite{google2026tpu8}.

Boardfly is a good answer to the question it was asked. We observe that it does
not answer the next one. A complete graph on $G$ vertices requires degree
$G-1$; a group has $40$ optical ports; a $400{,}000$-chip machine has $12{,}500$
groups. The mechanism that gives Boardfly its diameter is precisely the
mechanism that prevents it from scaling, and the standard remedy---stack another
hierarchy level---costs four chip hops per level.

\paragraph{Contributions.}
\begin{itemize}
\item We catalogue every shipped TPU interconnect as a graph and audit each
against the Moore bound (\S\ref{sec:moore}). At TPU-class radix ($\Delta
\approx 7$) a diameter-$7$ graph could in principle hold $391{,}910$ chips;
Boardfly spends diameter $7$ on $1{,}152$.
\item We give \textbf{Shiftfly} (\S\ref{sec:design}), which replaces only
Boardfly's global tier with a generalized Kautz digraph on the existing optical
circuit switch, at an unchanged port budget. We prove its diameter, its
table-free routing rule, and a criterion under which requests for the same
content merge inside the fabric (\S\ref{sec:analysis}).
\item We identify a port-accounting error that would silently double Shiftfly's
hardware if the comparison were done naively, and we prove the correct
accounting (\S\ref{sec:ports}).
\item We evaluate both designs at matched optical cost with a dependency-free
simulator whose Boardfly model reproduces the published diameter of~$7$
(\S\ref{sec:eval}). Shiftfly loses at pod scale and wins beyond it.
\item We report a \textbf{negative result} (\S\ref{sec:sharing}): the redundancy
argument that motivated the design is dominated by locality-aware placement,
which is topology-agnostic. We also identify a metric inversion that would have
concealed this had we reported the natural efficiency ratio alone.
\item We measure whether the design can be \textbf{operated}
(\S\ref{sec:ops})---field replacement, control-plane state, slice allocation and
growth---rather than assuming it. Field replacement ties; slice allocation is a
genuine regression that we quantify and bound.
\end{itemize}

All code, proofs and artefacts are public~\cite{krause2026shiftfly}. No
proprietary information is used anywhere in this work; every TPU figure is drawn
from public announcements, and quantities we infer rather than read are marked
as inferred.

%======================================================================
\section{Background: TPU interconnects as graphs}
\label{sec:background}

\subsection{The torus era}

\begin{definition}[$k$-ary $n$-cube]
The torus $T(k_1,\dots,k_n)$, with every $k_i \ge 2$, is the
Cayley graph of
\[
  \Z_{k_1}\times\cdots\times\Z_{k_n}
  \quad\text{with generators}\quad
  \{\pm e_1,\dots,\pm e_n\}.
\]
\end{definition}

\begin{proposition}
\label{prop:torus}
$T(k_1,\dots,k_n)$ has $N=\prod_i k_i$ vertices, is $2n$-regular when every
$k_i>2$, and $\diam T = \sum_{i}\lfloor k_i/2\rfloor$. A cubic torus therefore
has $\diam T = \Theta(N^{1/n})$.
\end{proposition}

\begin{proof}
In a Cayley graph of an abelian group with this generating set, distance is the
sum of per-coordinate cyclic distances; the cyclic distance in $\Z_k$ is
maximised at $\lfloor k/2\rfloor$, and the maxima are simultaneously
attainable.
\end{proof}

Table~\ref{tab:generations} lists the generations. Pod shapes for v5p, v7 and
v8t are inferred from published chip counts and marked accordingly; the
resulting diameters follow from Proposition~\ref{prop:torus}.

\begin{table}[t]
\centering\small
\begin{tabular}{lrrr}
\toprule
Generation & Chips & $\Delta$ & $D$ \\
\midrule
v2 / v3 (2D torus) & 256 / 1{,}024 & 4 & 16 / 32 \\
v4 (3D torus, OCS between cubes) & 4{,}096 & 6 & 24 \\
v5e / v6e (2D torus) & 256 & 4 & 16 \\
v5p (3D torus)$^\dagger$ & 8{,}960 & 6 & ${\approx}32$ \\
v7 Ironwood (3D torus)$^\dagger$ & 9{,}216 & 6 & ${\approx}32$ \\
v8t (3D torus + Virgo)$^\dagger$ & 9{,}600 & 6 & ${\approx}33$ \\
\textbf{v8i Boardfly} & \textbf{1{,}152} & \textbf{7} & \textbf{7} \\
\bottomrule
\end{tabular}
\caption{TPU interconnects as graphs. $^\dagger$Pod shape inferred from
published chip counts.}
\label{tab:generations}
\end{table}

\subsection{Boardfly}

Google describes Boardfly as ``a building block of four fully connected chips
into a fully connected group of eight boards, with 36 of such groups fully
connected into a TPU~8i pod''~\cite{google2026eighthgen}. A building block holds
$4$ chips; a group holds $8$ building blocks, i.e.\ $32$ chips; a pod holds $36$
groups, i.e.\ $1{,}152$ chips. Each building block exposes $16$ external links
of which $11$ serve the group in copper, leaving $5$ per building block and
hence $40$ optical ports per group. Within a pod the $36$ groups are fully
connected through the Apollo optical circuit switch~\cite{poutievski2022jupiter},
consuming $35$ of those $40$.

\begin{lemma}[Diameter decomposition]
\label{lem:seven}
If the optical link between two groups terminates on one designated chip of one
designated building block on each side, the chip-level pod diameter is $7$.
\end{lemma}

\begin{proof}
Reaching a group's egress chip costs at most three hops: one inside the source
building block, one across copper to the building block owning the optical port,
one inside that block. The optical hop is one, and the destination side is
symmetric, giving $3+1+3=7$; the bound is attained.
\end{proof}

We adopt this accounting throughout. Traversing $g$ inter-group hops passes
through $g+1$ groups, so
\begin{equation}
\label{eq:chiphops}
\text{chip hops} = \iota\,(g+1)+g, \qquad \iota=3,
\end{equation}
which reproduces $7$ at $g=1$. Our simulated chip-level pod measures diameter
exactly $7$, which is the check that the port model is faithful
(\S\ref{sec:validation}).

%======================================================================
\section{Motivation}
\label{sec:moore}

\subsection{The Moore-bound audit}

\begin{theorem}[Moore bound~\cite{hoffman1960,miller2013degdiam}]
A graph of maximum degree $\Delta\ge3$ and diameter $D$ satisfies
$N \le 1+\Delta\sum_{i=0}^{D-1}(\Delta-1)^i$.
\end{theorem}

\begin{proof}
Breadth-first search from any vertex reaches at most $\Delta$ vertices at
distance $1$ and at most $\Delta(\Delta-1)^{i-1}$ at distance $i$, since each
newly reached vertex has at most $\Delta-1$ unused incident edges; every vertex
lies within $D$.
\end{proof}

Equality demands a Moore graph, and these exist only for $D=1$, for cycles, and
for $D=2$ with $\Delta\in\{3,7,57\}$~\cite{hoffman1960}. The bound is a ceiling,
so the informative quantity is the ratio to it. At $\Delta=7$:
\[
\begin{aligned}
\text{Moore}(7,4) &= 1{,}814, &\quad \text{Moore}(7,6) &= 65{,}318,\\
\text{Moore}(7,5) &= 10{,}886, &\quad \text{Moore}(7,7) &= 391{,}910.
\end{aligned}
\]

\begin{observation}
Boardfly attains diameter $7$ at $1{,}152$ chips, where $\Delta=7$ permits
diameter $4$ at that order---a factor of ${\approx}1.75$. More strikingly,
$\Delta=7$ admits diameter $7$ at almost exactly $400{,}000$ chips. Boardfly
spends, at pod scale, the diameter budget that would in principle serve a
machine two orders of magnitude larger.
\end{observation}

We are careful not to overclaim. Moore-optimal graphs of large diameter are not
known, and the bound ignores bisection, cable length and fault tolerance---all
of which Boardfly buys with its slack. The observation motivates inspection; it
is not by itself an argument.

\subsection{The structural obstruction}

The argument is this. Boardfly's global tier is the complete graph $K_{36}$ on
groups, so reaching $G$ groups costs $G-1$ optical ports per group. At
$400{,}000$ chips, i.e.\ $G=12{,}500$ groups of $32$, this is $12{,}499$ ports
against $40$ available. Full connectivity does not scale, another hierarchy
level must be introduced, and by~\eqref{eq:chiphops} every added level costs
four chip hops. The property that gives Boardfly its diameter is the property
that bounds its size.

%======================================================================
\section{Design}
\label{sec:design}

\begin{figure*}[t]
\centering
\includegraphics[width=0.88\textwidth]{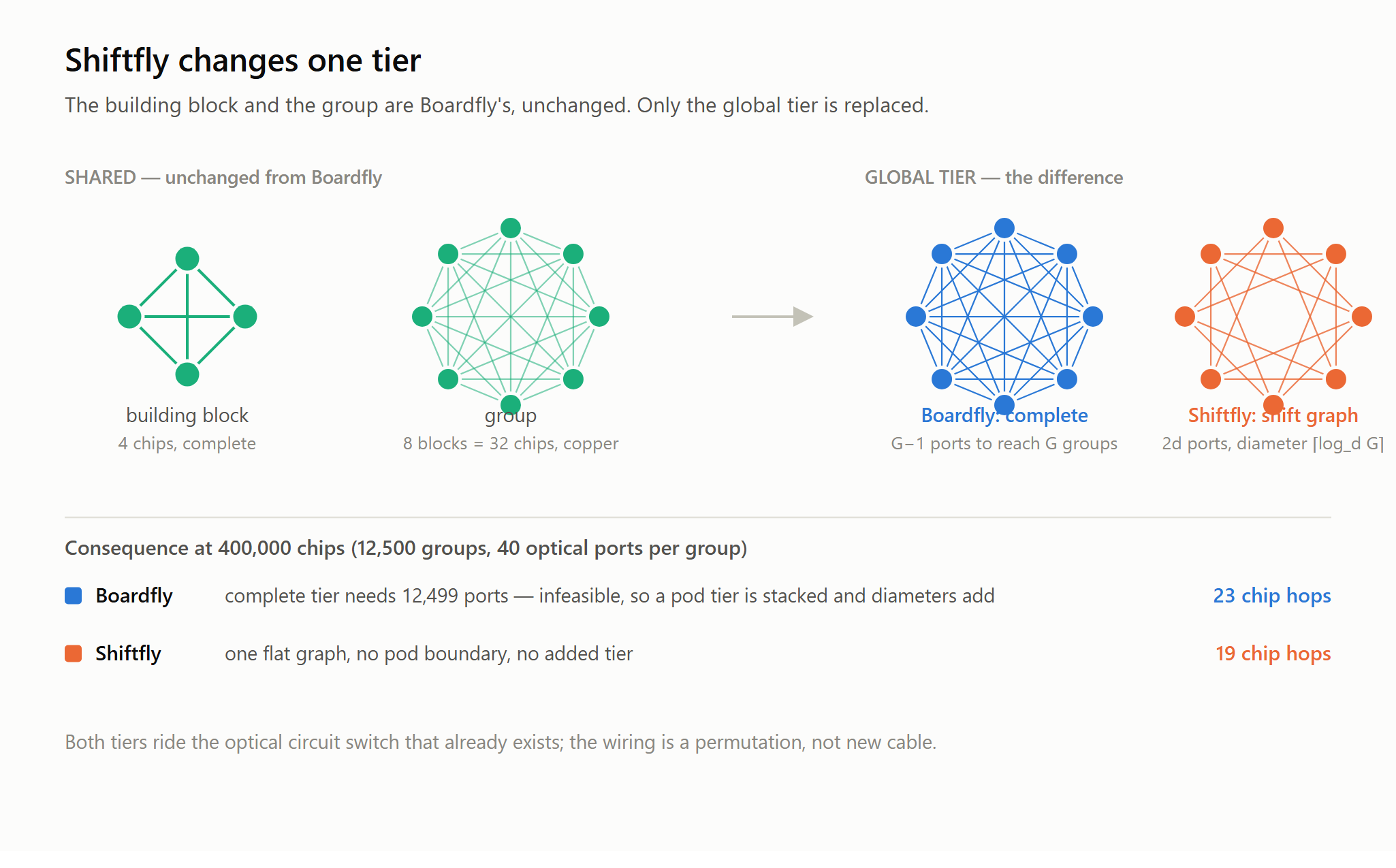}
\Description{Four panels. The left two, labelled shared and unchanged, show a
building block of four fully connected chips and a group of eight such blocks.
The right two contrast the global tier: Boardfly's complete graph, needing G
minus one ports to reach G groups, against Shiftfly's sparser shift graph
needing 2d ports at diameter log-base-d of G. A footer states the consequence
at 400,000 chips: 23 chip hops for Boardfly against 19 for Shiftfly.}
\caption{Shiftfly changes exactly one tier. The building block and the group are
Boardfly's, unchanged; only the global tier is replaced, and it rides the
optical circuit switch that already exists.}
\label{fig:hierarchy}
\end{figure*}

\begin{definition}[Kautz digraph~\cite{kautz1968}]
For $d\ge2$, $D\ge1$, the vertices of $K(d,D)$ are words $u_1\cdots u_D$ over an
alphabet of size $d+1$ with $u_i\ne u_{i+1}$, and there is an arc
$u \to u_2\cdots u_D\alpha$ for each $\alpha\ne u_D$.
\end{definition}

\begin{definition}[Shiftfly]
Retain Boardfly's building block ($4$ chips, complete) and group ($8$ building
blocks, complete; $32$ chips). Replace the complete global tier by a Kautz
digraph on the groups, realised as a fixed permutation installed on the optical
circuit switch. Where the group count is not of the form $(d+1)d^{D-1}$, use the
Imase--Itoh digraph~\cite{imase1981,imase1983} on $\Z_n$ with arcs
$i\mapsto -di-r \bmod n$, $r=1,\dots,d$, which exists at every order.
\end{definition}

Figure~\ref{fig:hierarchy} shows what changes. Nothing below the global tier is
touched: the copper, the packaging and the $32$-chip group are Boardfly's. This
is deliberate---those tiers are cheap, short-reach and well designed, and the
obstruction of \S\ref{sec:moore} lives entirely above them.

\subsection{Port accounting}
\label{sec:ports}

The comparison is worthless unless both designs are held to the same optical
budget, and the natural reading of ``out-degree $d$'' understates Shiftfly's
cost by a factor of two.

\begin{proposition}
\label{prop:ports}
In $K(d,D)$ with $D\ge2$, a vertex is both a successor and a predecessor of
another exactly when $u_1=u_D$, and then in exactly one neighbour. There are
$d(d+1)$ such vertices out of $(d+1)d^{D-1}$. Undirected degree is therefore
$2d$, or $2d-1$ on that thin set.
\end{proposition}

\begin{proof}
$v$ is a successor and predecessor of $u$ iff $u_2\cdots u_D\alpha =
\beta u_1\cdots u_{D-1}$ for admissible $\alpha,\beta$. Matching coordinates
forces $\beta=u_2$, $u_i = u_{i+2}$ for all valid $i$, and $\alpha=u_2$; with
the no-repeat constraint this reduces to $u_1=u_D$ and determines $\alpha$
uniquely. Counting words with $u_1=u_D$ gives $d(d+1)$.
\end{proof}

\begin{corollary}
Out-degree $d$ costs $2d$ bidirectional ports. Matching Boardfly's $40$ optical
ports per group means $d=20$, not $d=40$.
\end{corollary}

Every result below uses $d=20$, and our artefact asserts that the two designs'
edge counts agree to within $5\%$ at each scale.

\subsection{Routing}

\begin{figure*}[t]
\centering
\includegraphics[width=0.88\textwidth]{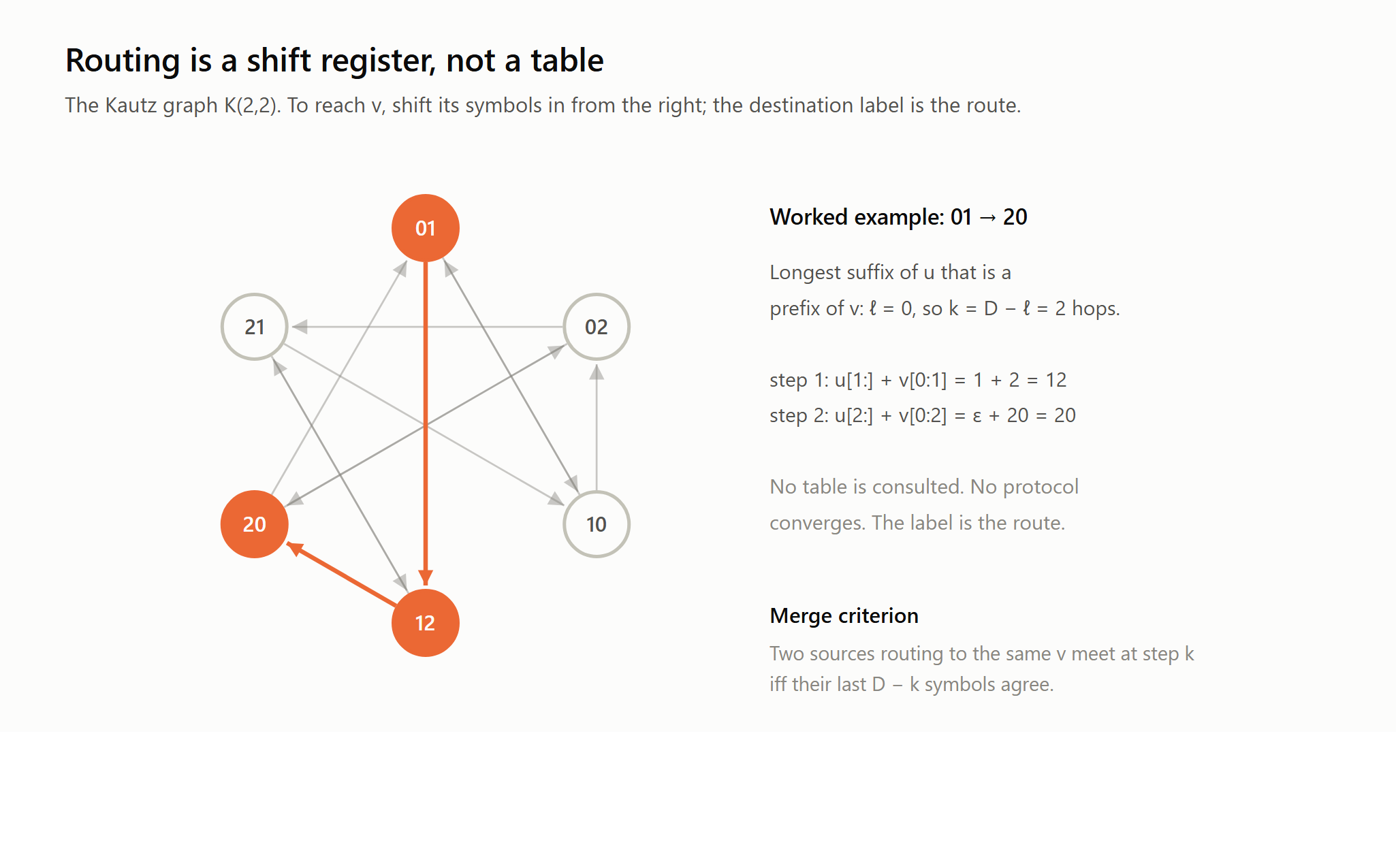}
\Description{The six vertices of the Kautz graph K(2,2), labelled 01, 02, 10,
12, 20 and 21, arranged in a ring with directed arcs between them. The route
from 01 to 12 to 20 is highlighted. Alongside, the arithmetic: the longest
suffix of the source that prefixes the destination has length zero, so the
route takes two hops, appending the destination symbols one per step.}
\caption{Shift routing on $K(2,2)$, with the two-hop route $01\to12\to20$. The
destination label is the route: no table is consulted and no protocol
converges.}
\label{fig:routing}
\end{figure*}

Routing is a shift register (Figure~\ref{fig:routing}); the rule and its
optimality are Theorem~\ref{thm:shift}.

%======================================================================
\section{Analysis}
\label{sec:analysis}

\begin{theorem}[Order, degree, diameter~\cite{kautz1968}]
\label{thm:kautz}
$K(d,D)$ has $N=(d+1)d^{D-1}$ vertices, in- and out-degree $d$, and diameter
exactly $D$.
\end{theorem}

\begin{proof}
There are $d+1$ choices for $u_1$ and $d$ for each later symbol, giving the
order; out-degree is the number of $\alpha\ne u_D$. For the diameter, applying
the arc map $D$ times replaces every symbol of $u$ by one of $v$, so
$\dist(u,v)\le D$, and $D$ is attained by any pair sharing no admissible
overlap.
\end{proof}

Against the directed Moore bound $N\le(d^{D+1}-1)/(d-1)$, $K(d,D)$ is within a
factor $1+O(1/d)$: it is essentially optimal among digraphs, which is the reason
to prefer it to de~Bruijn here.

\begin{theorem}[Shift routing]
\label{thm:shift}
Let $\ell$ be the length of the longest suffix of $u$ that is a prefix of $v$,
and $k=D-\ell$. Then
\[
  w^{(j)} \;=\; u_{j+1}\cdots u_D\, v_{\ell+1}\cdots v_{\ell+j},
  \qquad j=0,\dots,k,
\]
is a walk from $u$ to $v$ of length $k$, and $k=\dist(u,v)$.
\end{theorem}

\begin{proof}
$w^{(0)}=u$, and $w^{(j+1)}$ is obtained from $w^{(j)}$ by deleting the leading
symbol and appending $v_{\ell+j+1}$, which is the arc map. At $j=k$ we have
$w^{(k)} = u_{k+1}\cdots u_D\,v_{\ell+1}\cdots v_D$, and $u_{k+1}\cdots u_D$ has
length $\ell$ and equals $v_1\cdots v_\ell$ by definition of $\ell$; hence
$w^{(k)}=v$. Each $w^{(j)}$ is a legal vertex because $u_D=v_\ell \ne
v_{\ell+1}$. Minimality follows because reaching $v$ in $k'$ steps requires the
length-$(D-k')$ suffix of $u$ to equal the corresponding prefix of $v$, and
$\ell$ is the largest such overlap.
\end{proof}

Theorem~\ref{thm:shift} is the operational point: the destination label is the
route. There is no forwarding table to size, no routing protocol to converge and
no state to reconcile after a reconfiguration, at any scale.

\begin{theorem}[Merge criterion]
\label{thm:merge}
Let $u,u'$ route to a common $v$ under Theorem~\ref{thm:shift} with equal
overlap. Their walks occupy the same vertex at step $j$ if and only if $u$ and
$u'$ agree in their last $D-j$ symbols. The first merge is therefore at
$j^\ast = D-\lcs(u,u')$, where $\lcs$ is the longest common suffix.
\end{theorem}

\begin{proof}
By Theorem~\ref{thm:shift}, $w^{(j)}$ depends on the source only through
$u_{j+1}\cdots u_D$; the appended symbols are a function of $v$ alone. Equality
of $w^{(j)}$ and $w'^{(j)}$ is thus equivalent to equality of those suffixes.
\end{proof}

Theorem~\ref{thm:merge} says that redundant requests collapse
\emph{structurally}: once two requests for the same item meet, the second may be
served from an in-network cache at the meeting point, with no directory, no
coordination and no protocol. It is a statement about the routing algebra, and
it is the property that motivated this design.

\begin{proposition}[Random labels give almost no merging]
\label{prop:random}
Under a uniformly random labelling,
$\Pr[\lcs(u,u')\ge j]=d^{-j}(1+o(1))$, whence
$\E[\lcs]=\tfrac{1}{d-1}(1+o(1))$ and $\E[j^\ast]=D-\tfrac{1}{d-1}(1+o(1))$.
\end{proposition}

At $d=20$ this is $\E[\lcs]\approx0.053$: essentially no merging before the
destination. \textbf{The merge property is not free.} It must be purchased by
choosing the labelling so that groups likely to request the same content are
suffix-close---an embedding of the observed sharing hypergraph into the suffix
tree. \S\ref{sec:sharing} measures what that is worth, and the answer is
sobering.

\begin{theorem}[Arbitrary order~\cite{imase1981,reddy1980}]
The Imase--Itoh digraph on $\Z_n$ has out-degree $d$ and diameter at most
$\lceil\log_d n\rceil$ for every $n\ge2$.
\end{theorem}

\begin{proof}[Proof sketch]
Iterating the arc map $k$ times sends $i$ to $(-d)^k i-\sum_j r_j(-d)^{k-j}$.
As the $r_j$ range over $\{1,\dots,d\}$ the reachable set after $k$ steps has
size $\min(d^k,n)$, and covers $\Z_n$ once $d^k\ge n$.
\end{proof}

This removes a quantisation that would otherwise strand capacity: string Kautz
graphs exist only at orders $(d+1)d^{D-1}$, and a machine is whatever size it
is.

\begin{proposition}[Connectivity~\cite{imase1983}]
$K(d,D)$ has vertex connectivity $d$, the maximum possible at that degree.
\end{proposition}

Thin global tiers are the usual fragility of hierarchical fabrics; maximal
connectivity is the relevant defence.

%======================================================================
\section{Evaluation}
\label{sec:eval}

\subsection{Methodology}

Both designs receive $40$ optical ports per group and identical $32$-chip group
internals, so all comparisons are at matched optical cost
(Proposition~\ref{prop:ports}). Distances are computed exactly by breadth-first
search where the instance permits and by sampling from $120$--$256$ random
sources otherwise; sampled diameters are lower bounds and are labelled as
sampled. Spectral gaps are power-iteration estimates of $1-\lambda_2$ of the
normalised adjacency operator, reported as estimates. The simulator is
stdlib-only and regenerates every figure and table in this paper from one
script~\cite{krause2026shiftfly}.

Boardfly beyond one pod is modelled generously: the $5$ optical ports per group
left after intra-pod full connectivity are wired as a random graph over pods,
which at $347$ pods yields a low-diameter pod-level graph. We are not
handicapping the incumbent.

\subsection{Model validation}
\label{sec:validation}

Our chip-level Boardfly pod has $1{,}152$ vertices, maximum degree $7$, mean
degree $5.84$, and \textbf{measured diameter exactly $7$}, matching the
published figure and Lemma~\ref{lem:seven}. Every conclusion below rests on this
check.

\subsection{Scaling}

\begin{figure}[t]
\centering
\includegraphics[width=\columnwidth]{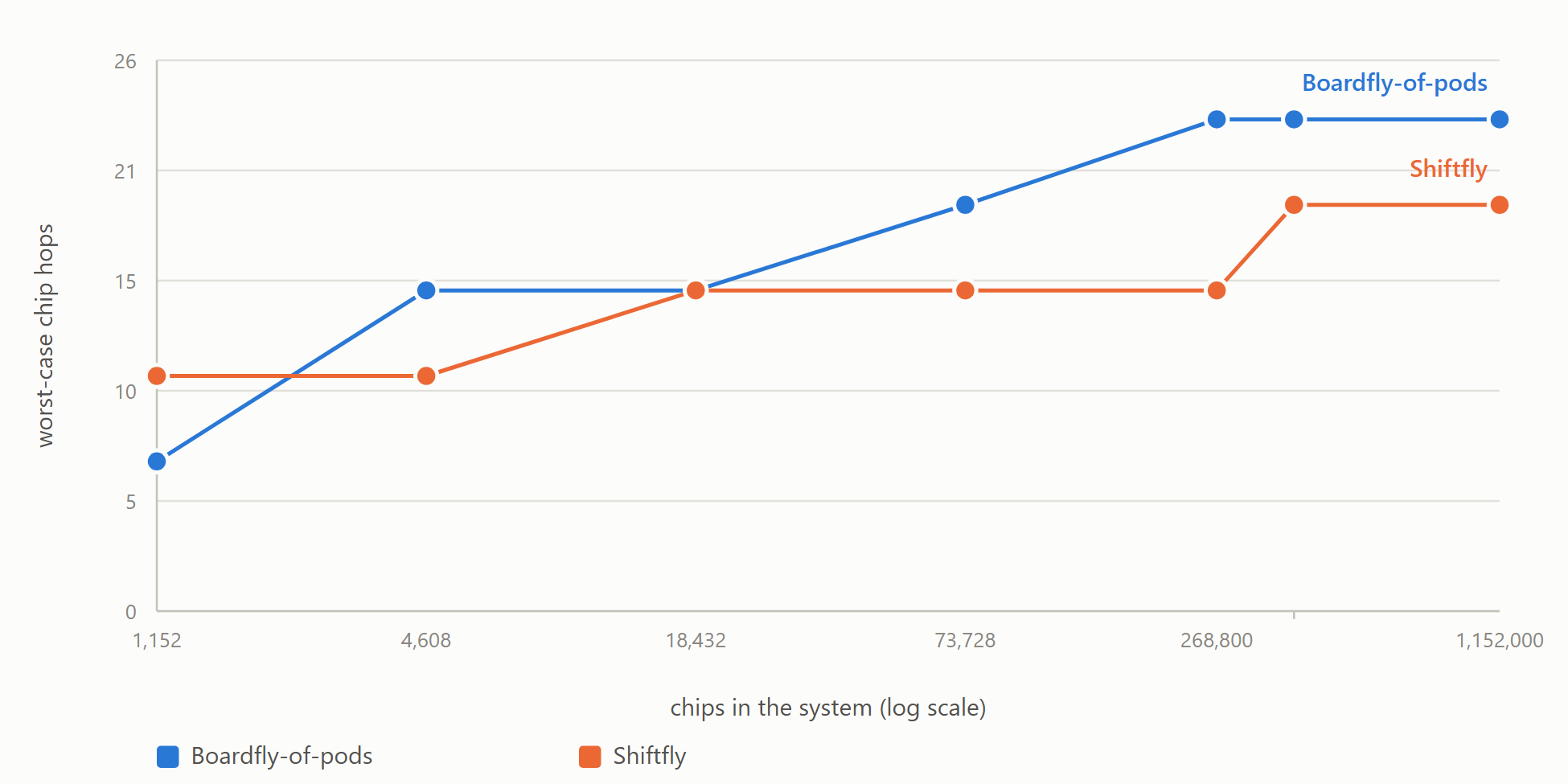}
\Description{Worst-case chip hops against system size on a logarithmic axis,
from 1,152 to 1,152,000 chips. Boardfly starts lower at 7 hops and rises in
steps to 23. Shiftfly starts at 11 and rises to 19. The two curves cross near
4,600 chips, after which Shiftfly is always at or below Boardfly.}
\caption{Chip-level worst-case distance at equal optical cost. Boardfly is
better at pod scale; the crossing is the result.}
\label{fig:diam}
\end{figure}

\begin{table}[t]
\centering\small
\begin{tabular}{rrr}
\toprule
Chips & Boardfly & Shiftfly \\
\midrule
1{,}152 & \textbf{7} & 11 \\
4{,}608 & 15 & \textbf{11} \\
18{,}432 & 15 & 15 \\
73{,}728 & 19 & \textbf{15} \\
268{,}800 & 23 & \textbf{15} \\
400{,}000 & 23 & \textbf{19} \\
1{,}152{,}000 & 23 & \textbf{19} \\
\bottomrule
\end{tabular}
\caption{Chip-level diameter at $40$ optical ports per group.}
\label{tab:scale}
\end{table}

Table~\ref{tab:scale} and Figure~\ref{fig:diam} give the result.
\textbf{Boardfly is strictly better at the scale it was designed for}: $7$
against $11$ chip hops at $1{,}152$ chips, because a complete global tier
crosses in one hop and a shift tier needs three. The curves cross near $4{,}600$
chips. At the $400{,}000$-chip target the reduction is $23\to19$, or $17\%$; the
widest gap in the sweep is $23\to15$ at $268{,}800$ chips, where Boardfly has
just paid for another hierarchy level and Shiftfly has not yet needed another
symbol. Both curves are step functions, so the advantage is not monotone.

\subsection{Expansion}

\begin{figure}[t]
\centering
\includegraphics[width=\columnwidth]{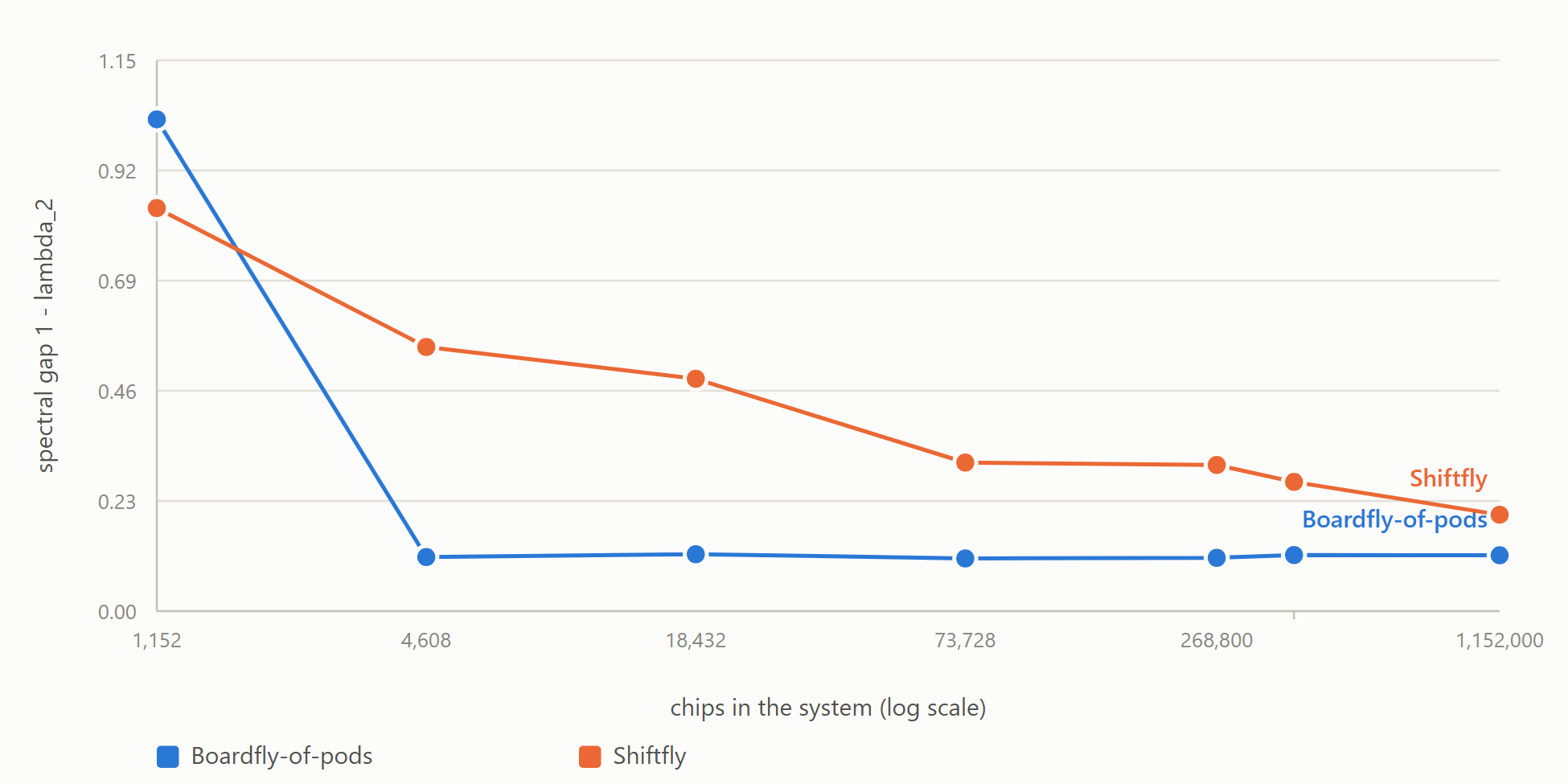}
\Description{Spectral gap against system size on a logarithmic axis. Both
curves fall as the machine grows, but Boardfly's collapses sharply once it
exceeds a single pod and then stays near zero, while Shiftfly's declines
gradually and remains roughly 2.7 times higher across the tested range.}
\caption{Spectral gap of the group graph. A hierarchy of near-cliques joined by
few links expands poorly whatever its diameter.}
\label{fig:gap}
\end{figure}

Shiftfly's spectral gap is roughly $2.7\times$ Boardfly's across the tested
range (Figure~\ref{fig:gap}). In hindsight this is unsurprising and it is worth
stating why, because it cuts against the usual intuition about de~Bruijn-family
graphs. Those families are known to have poor bisection relative to
\emph{expanders}~\cite{bermond1989debruijn}. Boardfly beyond one pod is not an
expander: it is a set of near-cliques joined by a thin random tier, which is
close to the worst case for spectral expansion. The comparison therefore favours
Shiftfly here, and would not against a flat high-radix fabric.

\subsection{Sharing: a negative result}
\label{sec:sharing}

\begin{figure}[t]
\centering
\includegraphics[width=\columnwidth]{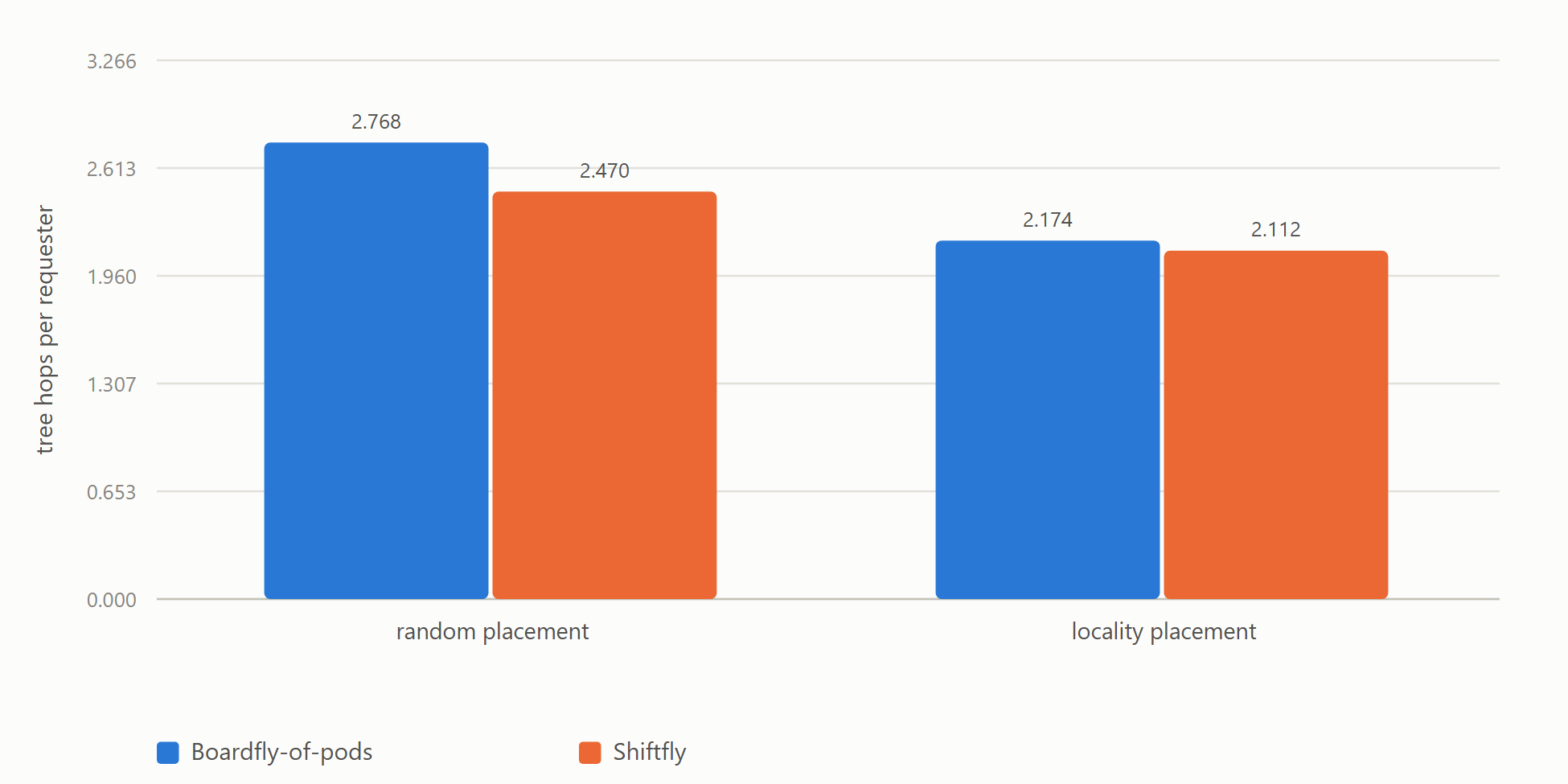}
\Description{Four bars in two groups showing tree hops per requester. Under
random placement Boardfly costs 2.768 and Shiftfly 2.470. Under locality-aware
placement both fall sharply, to 2.174 and 2.112 respectively. The drop from
changing placement is far larger than the remaining gap between the two
topologies.}
\caption{Cost of serving one shared item to $32$ agents, in tree hops per
requester. Placement does most of the work, on both topologies.}
\label{fig:sharing}
\end{figure}

Agentic inference fans one query into many concurrent rollouts that share a
prompt prefix, tool definitions and retrieved context, and therefore request the
same content~\cite{zheng2024sglang,qin2025mooncake,wang2023selfconsistency}.
Theorem~\ref{thm:merge} says a shift-routed fabric collapses such requests
structurally. We set out to quantify that advantage. It is small.

\begin{definition}
For home vertex $h$ and requester set $S$, the \emph{unicast cost} is
$U=\sum_{v\in S}\dist(v,h)$ and the \emph{tree cost} $T$ is the number of edges
in the union of shortest paths from $S$ to $h$---what a fabric with in-network
replication pays. The \emph{efficiency} is $\eta=U/T$.
\end{definition}

\begin{table}[t]
\centering\small
\begin{tabular}{lrrrr}
\toprule
& \multicolumn{2}{c}{$T$ per requester} & \multicolumn{2}{c}{merge depth}\\
\cmidrule(lr){2-3}\cmidrule(lr){4-5}
Placement & Boardfly & Shiftfly & Boardfly & Shiftfly\\
\midrule
random & 2.768 & 2.470 & 2.868 & 2.772 \\
locality-aware & 2.174 & \textbf{2.112} & 2.336 & \textbf{2.152} \\
\midrule
saving & $-21.5\%$ & $-14.5\%$ & & \\
\bottomrule
\end{tabular}
\caption{Shared-content cost. Lower is better. Placement dominates.}
\label{tab:sharing}
\end{table}

Table~\ref{tab:sharing} and Figure~\ref{fig:sharing} report a workload of $32$
agents per query over $64$ affinity clusters with $85\%$ scheduler locality and
Zipf-distributed item popularity, on $8{,}400$ groups. Locality-aware placement
removes $21.5\%$ of tree cost on Boardfly and $14.5\%$ on Shiftfly.
\textbf{Shiftfly's residual advantage at matched placement is $2.9\%$}, with
merge depth $7.9\%$ lower. The great majority of the achievable saving on shared
content therefore comes from \emph{placement}, which is topology-agnostic and
equally available to Boardfly.

\begin{remark}[A metric inversion]
\label{rem:trap}
$\eta$ must not be optimised alone. It is a ratio, so a topology is rewarded for
placing requesters \emph{far} from the home vertex provided their paths overlap;
a network in which everything is one hop away scores the worst possible $\eta=1$
while paying the least. In our measurements $\eta$ ranks Boardfly \emph{higher}
($1.431$ against $1.364$) while absolute tree cost ranks it \emph{lower}
($2.174$ against $2.112$). Had we reported $\eta$ alone---the natural reading of
a sharing objective---we would have drawn the opposite conclusion.
\end{remark}

%======================================================================
\section{Deployment and operations}
\label{sec:ops}

A topology that is better on paper and worse to run is not better. Shiftfly is
proposed as a drop-in replacement for one tier of a machine that already exists,
so the bar is concrete: servicing, allocating and growing it must be at least as
convenient as Boardfly. This section measures four such questions on a
$12{,}500$-group ($400{,}000$-chip) configuration. Three of the four favour
Shiftfly or tie; one does not, and we give it its own subsection.

\subsection{Replacing a failed group}

In both designs a replacement unit \emph{inherits the identity} of the unit it
replaces, and the optical circuit switch re-points that identity's fibres at the
new hardware. The cost is therefore just the degree, and the same $40$-port
budget binds both: $40$ circuits re-established, $40$ peers disturbed, in each
case. Neither design is exactly regular---Shiftfly loses a port on the thin set
of vertices carrying a bidirectional pair (Proposition~\ref{prop:ports}),
Boardfly wherever an inter-pod stub fails to pair---but the mean costs agree to
within a few percent.

\textbf{Field replacement is therefore a tie, by construction rather than by
luck.} Nothing about a shift-routed tier makes a chip, tray or group harder to
pull. What differs is that Shiftfly need not \emph{look anything up}: the
neighbour set of a label is a shift of that label, so the technician's work
order and the control plane's reconciliation are the same in both designs while
the state behind them is not.

\subsection{Control-plane state}

\begin{table}[t]
\centering\small
\begin{tabular}{lr}
\toprule
Design & State for the global wiring \\
\midrule
Boardfly & $550{,}000$ bits \\
Shiftfly & $28$ bits \\
\bottomrule
\end{tabular}
\caption{State the control plane must hold to derive which groups are wired to
which, at $12{,}500$ groups.}
\label{tab:state}
\end{table}

Boardfly's intra-pod tier is derivable---a complete graph needs only pod
membership---but its inter-pod tier has no closed form and must be tabulated.
Shiftfly's entire global tier follows from the pair $(n,d)$, because
Theorem~\ref{thm:shift} makes adjacency arithmetic. Table~\ref{tab:state} gives
the resulting difference of four orders of magnitude.

We do not oversell this. Boardfly could adopt a \emph{structured} inter-pod tier
and recover most of the gap---but a structured global tier over an OCS is
essentially the proposal of this paper.

\subsection{Slice allocation: where Shiftfly is worse}
\label{sec:slices}

\begin{table}[t]
\centering\small
\begin{tabular}{rrccc}
\toprule
Slice & Chips & Boardfly & SF naive & SF instantiated \\
\midrule
16 & 512 & \textbf{1} & disconn. & \textbf{1} \\
36 & 1{,}152 & \textbf{1} & disconn. & 2 \\
128 & 4{,}096 & 3 & disconn. & \textbf{2} \\
512 & 16{,}384 & 4 & disconn. & \textbf{3} \\
2{,}048 & 65{,}536 & 4 & disconn. & \textbf{3} \\
\bottomrule
\end{tabular}
\caption{Group-level diameter of an allocated slice. An arbitrary induced subset
of a shift graph is disconnected at every size tested.}
\label{tab:slices}
\end{table}

Google allocates slices of a pod to jobs, and slices must be isolated and
internally well connected. Here Boardfly has a structural advantage that
Shiftfly cannot match directly: \emph{any} subset of a complete tier is itself
complete, so a slice of up to $36$ groups is carved out at zero cost and with
diameter $1$, requiring no reconfiguration at all.

The corresponding naive operation on Shiftfly---take an arbitrary subset and use
whichever links fall inside it---fails badly. Table~\ref{tab:slices} shows the
induced subgraph is \textbf{disconnected at every size we tested}. Slices cannot
be carved out of a shift graph.

They can, however, be \emph{instantiated}. Because the global tier is a
permutation installed on an optical circuit switch, and because the Imase--Itoh
construction exists at every order, a slice of $m$ groups can be given its own
correctly sized shift permutation. Every slice is then a fabric in its own right
rather than a fragment of a larger one, with the diameter guarantee
$\lceil\log_d m\rceil$ intact. On those terms Shiftfly slices are \emph{better}
than Boardfly's from $128$ groups ($4{,}096$ chips) upward, where Boardfly's own
slice has had to become hierarchical.

The cost is one OCS reconfiguration per allocation, which Boardfly avoids within
a single pod. We regard this as acceptable, but only because it matches how the
machine is already operated: MEMS reconfiguration takes milliseconds to seconds
and happens at job-scheduling
time~\cite{poutievski2022jupiter,jouppi2023tpuv4}, against job lifetimes of
minutes to days, and any slice spanning more than one pod obliges Boardfly to
establish inter-pod circuits anyway. A deployment that allocates many small,
short-lived slices inside one pod should keep Boardfly.

\subsection{Growth and software}

Boardfly is defined at multiples of a pod; crossing $36$ groups introduces a
hierarchy level rather than widening an existing one. Shiftfly is defined at
\emph{every} order, so a partially populated machine of $n$ groups is a valid
fabric with the diameter its size warrants, and growth re-installs a
permutation rather than restructuring a tier.

The software change is bounded by construction, because the building block and
the group are untouched. Every intra-group collective---including the
latency-critical small reductions the collectives engine
accelerates---sees an unchanged fabric. What changes is the inter-group
collective schedule, which is topology-specific in any case, and the routing
layer, which gets \emph{simpler}: there is no forwarding table to size, no
protocol to converge, and no state to reconcile after a reconfiguration.

%======================================================================
\section{Discussion and limitations}
\label{sec:limits}

\paragraph{Boardfly is right at pod scale.} At $1{,}152$ chips a complete global
tier is simply the correct answer, and our own measurements say so. Nothing here
argues that TPU~8i should have been built differently.

\paragraph{Slices must be instantiated, not carved.} An arbitrary induced subset
of a shift graph is disconnected (\S\ref{sec:slices}). This is the one
operational regression, it is real, and a fleet dominated by small short-lived
in-pod slices should not adopt Shiftfly.

\paragraph{Bisection.} De~Bruijn and Kautz families have asymptotically smaller
bisection than a complete or expander tier~\cite{bermond1989debruijn}. We win
against a \emph{hierarchy} because hierarchies expand badly; against a flat
high-radix fabric we would not. All-to-all-dominated training is the adversarial
workload, and it is exactly the workload TPU~8t retains a torus and a flat
scale-out fabric for.

\paragraph{Cabling.} Shift edges are long and irregular, which is the standard
reason de~Bruijn networks are not built~\cite{bermond1989debruijn}. Our
mitigation is narrow and we do not generalise it: the global tier is
\emph{already} optically circuit-switched, so the wiring is a permutation the
switch installs rather than cable anyone routes. This argument is available to
an operator that owns an OCS layer~\cite{poutievski2022jupiter} and to
essentially nobody else.

\paragraph{Deterministic routing.} Theorem~\ref{thm:shift} yields a single path.
Load balance will require non-minimal
alternatives~\cite{valiant1982,jiang2009ugal}, and randomising the route partly
destroys the merge property of Theorem~\ref{thm:merge}. We have not quantified
that interaction, and it is the most important piece of missing work.

\paragraph{The switch is mechanical.} MEMS-based optical circuit switches
reconfigure on millisecond-to-second
timescales~\cite{poutievski2022jupiter,jouppi2023tpuv4}, so the Kautz
permutation is fixed at scheduling time. Label assignment is a scheduling
decision, not a runtime one.

\paragraph{The workload model is synthetic.} Cluster count, agent fan-out,
scheduler locality and Zipf exponent are plausible rather than measured; the
conclusions of \S\ref{sec:sharing} are only as good as those four numbers.

%======================================================================
\section{Related work}

Dragonfly~\cite{kim2008dragonfly} established the high-radix hierarchical form
that Boardfly follows, and HyperX~\cite{ahn2009hyperx} the flattened
alternative. The degree--diameter literature supplies our
ceiling~\cite{hoffman1960,miller2013degdiam}, and a line of recent topologies
pursues it directly: Slim~Fly from McKay--Miller--\v{S}ir\'a\v{n}
graphs~\cite{besta2014slimfly}, PolarFly from Erd\H{o}s--R\'enyi polarity
graphs~\cite{lakhotia2022polarfly}, and Bundlefly~\cite{lei2020bundlefly}. These
operate at \emph{router} radix, where a switch has tens of ports; Shiftfly
operates at accelerator radix, where the endpoint is the router and the port
count is single-digit per chip.

Kautz digraphs are due to Kautz~\cite{kautz1968}; arbitrary-order
generalizations to Reddy, Pradhan and Kuhl~\cite{reddy1980} and Imase and
Itoh~\cite{imase1981,imase1983}. The same shift algebra underlies the Koorde
distributed hash table~\cite{kaashoek2003koorde}, which is where we take the
content-addressing observation from; our contribution is to place it in the
physical topology rather than in an overlay. Google's optical circuit switching
infrastructure is described in Jupiter Evolving~\cite{poutievski2022jupiter} and
its use for TPU reconfiguration in TPU~v4~\cite{jouppi2023tpuv4}. Redundancy in
LLM serving is exploited above the network by prefix
caching~\cite{zheng2024sglang} and disaggregated KV
stores~\cite{qin2025mooncake}; we ask whether it can be exploited in the fabric,
and largely conclude that it is better handled by placement.

%======================================================================
\section{Conclusion}

Google's move from torus to Boardfly traded a $\Theta(N^{1/3})$ diameter for a
constant one, and was clearly correct for a $1{,}152$-chip inference pod. The
mechanism it used---full connectivity at every tier---does not extend, because
complete graphs need $N-1$ ports. Replacing only the global tier with a
generalized Kautz digraph keeps the diameter logarithmic in the group count at a
fixed port budget, yields table-free routing and native content addressing, and
measurably improves expansion, at the cost of being worse inside a single pod.

The redundancy argument that motivated the design did not survive its own
evaluation: most of the available saving on shared content comes from placement,
not from topology, and the natural metric for that objective inverts the
ranking. We regard reporting this as more useful than the alternative.

% Balances the final page's columns.  It emits "You have called \balance in
% second column" -- harmless, and it appears with or without this call, since
% acmart invokes balancing itself.  Removing it drops references from the
% final column, so it stays.
\balance
\bibliographystyle{ACM-Reference-Format}
\bibliography{shiftfly}

\end{document}